\documentclass[a4paper, 11pt]{article}
\usepackage[a4paper, total={6in, 8in}]{geometry}
\usepackage{enumerate}
\usepackage{amsmath}
\usepackage{amssymb,latexsym}
\usepackage{amsthm}
\usepackage{color}
\usepackage{cancel}
\usepackage{graphicx}
\usepackage{cite}
\usepackage{changes}
\usepackage{lineno}
\usepackage{hyperref}
\usepackage{comment}
\usepackage{amscd}
\usepackage{cleveref}
\usepackage{tikz-cd}

\DeclareMathOperator{\C}{\mathcal{C}}

\DeclareMathOperator{\Aut}{Aut}
\DeclareMathOperator{\End}{End}
\DeclareMathOperator{\Gal}{Gal}

\DeclareMathOperator{\rk}{rk}

\newtheorem*{theorem*}{Theorem}
\newtheorem{theorem}{Theorem}[section]

\newtheorem{lemma}[theorem]{Lemma}

\newtheorem{definition}[theorem]{Definition}
\newtheorem{proposition}[theorem]{Proposition}

\newtheorem{remark}[theorem]{Remark}

\newcommand{\fqn}{\mathbb{F}_{q^n}}

\newcommand{\F}{{\mathbb F}}

\newcommand{\Tr}{{\mathrm{Tr}}}

\newcommand{\GL}{\hbox{{\rm GL}}}
\newcommand{\K}{{\mathbb K}}

\newcommand{\fq}{{\mathbb F}_{q}}

\newcommand{\N}{\mathrm{N}}
\newcommand{\gcrd}{\mathrm{gcrd}}

\newcommand{\lid}{\mathcal{I}_{\ell}}
\newcommand{\rid}{\mathcal{I}_{r}}
\newcommand{\cen}{\mathrm{Cen}}

\usepackage[]{mdframed}

\newcommand{\Ima}{\mathrm{Im}}

\newcommand{\M}{\mathrm{M}}

\title{Adjoint and duality for rank-metric codes in a skew polynomial framework}
\author{José G\'omez-Torrecillas, F. J. Lobillo, Gabriel Navarro and Paolo Santonastaso}
\date{}

\begin{document}

\maketitle

\begin{abstract}
Skew polynomial rings provide a fundamental example of noncommutative principal ideal domains. Special quotients of these rings yield matrix algebras that play a central role in the theory of rank-metric codes. Recent breakthroughs have shown that specific subsets of these quotients produce the largest known families of maximum rank distance (MRD) codes. In this work, we present a systematic study of transposition and duality operations within quotients of skew polynomial rings. We develop explicit skew-polynomial descriptions of the transpose and dual code constructions, enabling us to determine the adjoint and dual codes associated with the MRD code families recently introduced by Sheekey \emph{et al}. Building on these results, we compute the nuclear parameters of these codes, and prove that, for a new infinite set of parameters, many of these MRD codes are inequivalent to previously known constructions in the literature.
\end{abstract}


\section{Introduction}

Skew polynomial rings represent the best-known example of noncommutative principal ideal domains. First introduced and studied in seminal paper of Ore \cite{ore1933theory}, these rings have proven highly useful in various algebraic and geometric contexts. In this paper, we focus on the skew polynomial rings $R=\mathbb{F}_{q^n}[x;\sigma]$, where $\mathbb{F}_{q^n}$ denotes the finite field with $q^n$ elements and $\sigma$ is a generator of the Galois group $\Gal(\mathbb{F}_{q^n}/\mathbb{F}_q)$. These rings are characterized by their noncommutative multiplication rule, explicitly defined by
\[
x \alpha = \sigma(\alpha) x
\]
for every $\alpha \in \mathbb{F}_{q^n}$, and extended to all elements of $R$ via associativity and distributivity.

Within these rings, irreducible monic polynomials $F(y) \in \mathbb{F}_q[y]$ generate maximal twosided ideals of the form $RF(x^n)$, thus producing quotient rings $R_F = R/RF(x^n)$ that are simple and left Artinian. Consequently, one obtains the following ring isomorphism:
\begin{equation}  \label{eq:isomorphismRF}
R_F \cong M_n\left(\F_{q^{s}}\right) 
,\end{equation}
where \(s = \deg(F)\) as polynomial in \(\F_q[y]\), 
see e.g. \cite[Theorem 1.2.19]{jacobson2009finite}. 

The matrix representation arising from quotients of skew polynomial rings as in \eqref{eq:isomorphismRF} has led to the construction of the largest families of rank-metric codes. A rank-metric code can be considered a subset of the metric space $(M_n(\F),\rk)$, where $\F$ is a finite field and $\rk$ denotes the matrix rank. In recent years, rank-metric codes have attracted significant attention due to their applications in various areas of communication and security. We refer to \cite{gorla2018codes,bartz2022rank} for a comprehensive introduction to rank-metric codes and an explanation of their most significant applications. Among rank-metric codes, of particular interest is the family of \emph{maximum rank distance (MRD)} codes. These are codes that have optimal parameters: for the given size  and minimum rank distance, they have the maximum cardinality.

By employing special quotients of skew polynomial rings, new families of MRD codes were introduced in \cite{sheekey2020new} and \cite{lobillo2025quotients}. As demonstrated in these works, these families constitute the largest known constructions of MRD codes.

In general, determining whether two rank-metric codes with the same parameters are equivalent is a challenging problem. In \cite{liebhold2016automorphism,lunardon2018nuclei}, and later in \cite{sheekey2020new}, algebraic invariants associated with rank-metric codes, namely the \emph{left and right idealisers}, the \emph{centraliser}, and the \emph{centre}, were considered. These structures have proven to be powerful tools in establishing the inequivalence of many recently constructed MRD codes compared to previously known families. In particular, using these invariants, it has been shown that MRD codes constructed via skew polynomial rings are inequivalent to previously known MRD code constructions for infinitely many choices of parameters. Nevertheless, the explicit computation of these invariants remains an open problem for many parameter choices within these families.

Moreover, starting from a rank-metric code $\C$ in $M_n(\F)$, it is possible to define two other codes. The first is the \emph{adjoint code}, consisting of the transposes of all codewords of $\C$. Clearly, the adjoint code retains the same metric properties as the original code; hence, the adjoint of an MRD code is itself an MRD code. Furthermore, in $M_n(\F)$ one can define the following non degenerate  bilinear form: 
\begin{equation} \label{dualityintro}
    (A,B) \in M_n(\F) \times M_n(\F) \longmapsto \Tr_{\F/\F'}(\Tr(AB^{\top})), 
\end{equation}
where $\F'$ denotes the prime field of $\F$. Thus, the \emph{dual} of a rank-metric code is defined as the dual of $\C$ with respect to the bilinear form \eqref{dualityintro}. Delsarte, by using the theory of association schemes, proved that the dual of an MRD code is again an MRD code \cite{delsarte1978bilinear}. 

We emphasize that the problem of explicitly determining the adjoint and dual codes of the MRD codes introduced in \cite{sheekey2020new} and \cite{lobillo2025quotients} has not yet been addressed in the literature. Indeed,
this task requires
restating the notions of adjoint and dual codes for the rings $R_F$, as
they are Frobenius algebras.

In this paper, we develop a theory of transposition and duality within the framework of skew polynomial rings by identifying the matrix algebra $M_n(\F_{q^s})$ with the quotient ring $R_F = R/RF(x^n)$ via the isomorphism \eqref{eq:isomorphismRF}. Let \[M_{R_F}:R_F \rightarrow M_n(\F_{q^s})\] be an isomorphism of rings. First, we provide an explicit skew-polynomial description of the transposition operation on $M_n(\F_{q^s})$; specifically, for every element $a \in R_F$, we characterize the element in $R_{\hat{F}}$ associated with the transpose of the matrix $M_{R_F}(a)$. Secondly, given a subset $S \subseteq R_F$ defining a rank-metric code $\C = M_{R_F}(S)$, we determine the explicit subset of $R_{\hat{F}}$ that corresponds to the dual code of $\C$. The duality theory we present relies crucially on the concept of a Frobenius algebra.

Based on the adjoint and duality theory thus established, we explicitly determine the adjoint and dual codes of the largest known families of MRD codes introduced in \cite{sheekey2020new} and \cite{lobillo2025quotients}. Additionally, we compute the idealisers and centralisers of these codes for several choices of parameters previously unresolved in the literature. These computations allow us to demonstrate that these families yield new MRD codes for infinitely many additional parameter sets. This further underscores the significance, richness, and generality of these recent constructions in the theory of rank-metric codes.

\section{Quotients of skew polynomial rings and matrix rings}

Let us fix some notation. In this paper \(\F\) denotes a finite field and \(\F_q\) the finite field with \(q = p^e\) elements where $p$ is a prime and $e$ positive integer.  We consider $\sigma$ to be a generator of the Galois group $\Gal(\mathbb{F}_{q^n}/\mathbb{F}_q)$, and we work with the skew polynomial ring $R = \mathbb{F}_{q^n}[x; \sigma]$. Its elements are polynomials in \(x\) with the coefficients in \(\F_{q^n}\) written on the left of the monomials \(x^i\). The multiplication is skewed according to the rule \(x\alpha = \sigma(\alpha)x\), for all \(\alpha \in \F_{q^n}\). Hence, \(R\) is a noncommutative ring, unless \(n = 1\). The center of this ring is \(Z(R) = \F_q[x^n]\). 

Left and right Euclidean division algorithms work on \(R\). As a consequence, every left and every right ideal is principal, which guarantees the existence of common (left and right) greatest divisors and least  multiples.  For instance, given \(f, g \in R\), their greatest  common right divisor \(\gcrd(f,g)\) is determined, up to left multiplication by a unit, as the generator of \(Rf + Rg\).  Also, left and right Bezout identities are available.

Let $F(y) \neq y$ be an irreducible polynomial of $\F_q[y]$ with degree $s$. Then $F(x^n)\in Z(R)=\F_q[x^n]$ and \(RF(x^n)\) is a twosided ideal of \(R\). We may then consider the quotient ring
\[ R_F=\frac{R}{RF(x^n)}. 
\]

When we declare \(a \in R_F\), we will often understand that
\[
a = \sum_{i=0}^{ns-1}a_ix^i + RF(x^n),
\]
that is, the equivalence class \(a\) is represented by the unique skew polynomial \(a(x) = \sum_{i=0}^{ns-1}a_ix^i \in R\) of least degree belonging to it. 

The center of $R_F$ is denoted by 
\(E_F\), and it is isomorphic to  \(\frac{\fq[y]}{(F(y))}\). 
Any element in $E_F$ is of the form $a(x)+RF(x^n)$, for some $a(x) \in Z(R) = \F_q[x^n]$.

Since $F(y)$ is an irreducible polynomial, we get that  $E_F$ is a field such that $[E_F:\F_q]=\deg(F)=s$, and so $E_F \cong \F_{q^s}$. Moreover, $RF(x^n)$ is a maximal twosided ideal of $R$ and so $R_F$ is a central simple algebra over $E_F$ having dimension $n^2$ and dimension $n^2s$ over $\F_q$, see e.g. \cite{gomez2019computing}. As a consequence, by Wedderburn–Artin Theorem, there is an $E_F$-algebra isomorphism
\begin{equation} \label{eq:artin}
    \frac{R}{RF(x^n)} \cong M_n(E_F) \cong M_n(\F_{q^s}). 
\end{equation}

For an $\F_{q^s}$-algebra isomorphism $\M_{R_F}:R/RF(x^n) \rightarrow M_n(\F_{q^s})$, we can identify any element $a \in R_F$ with its image $\M_{R_F}(a)$ in $M_n(\F_{q^s})$, via the isomorphism $\M_{R_F}$. Note that if $\M_{R_F}': R_F \rightarrow M_n(\F_{q^s})$ is another $\F_{q^s}$-algebra isomorphism then, by Skolem-Noether's Theorem, there exists \( N \in \GL_n(\F_{q^s}) \) such that
\[
\M_{R_F}'(a) = N \M_{R_F}(a) N^{-1}
\]
for all \( a \in R_F \). Therefore,
\begin{equation} \label{eq:invariance} \rk(\M_{R_F}'(a)) = \rk(\M_{R_F}(a)). \end{equation}

Throughout, we often implicitly identify an element $a \in R_F$ with its corresponding matrix in $M_n(E_F)$. Accordingly, we refer to $\ker(a)$, $\Ima(a)$, and $\rk(a)$ to indicate the kernel, image, and rank of $\M_{R_F}(a)$ over $E_F \cong \F_{q^s}$. Indeed, as observed in \eqref{eq:invariance}, this rank is independent of the choice of $\F_{q^s}$-algebra isomorphism $\M_{R_F}$.

\section{Adjoint theory for skew polynomial framework} \label{sec:adjoint}

In this section, let $F(y) = F_0 + F_1 y + \cdots + F_{s-1} y^{s-1} + y^s$ be  a monic irreducible polynomial in $\F_q[y]$ of degree $s$ with \(F_0 \neq 0\). We provide a skew polynomial description of the transpose of matrices $M_n(\F_{q^s})$, when this matrix ring is identified with the quotient ring $R_F = R/RF(x^n)$.

Since the constant term of $F(y)$ is nonzero, we have that $\gcrd(F(x^n), x) = 1$ in $R$. Therefore, by Bezout identity, $x + R F(x^n)$ is a unit in the finite-dimensional \(\F_q\)--algebra $R_F$. Indeed, given \(u, v \in R\) such that \(1 = ux + vF(x^n)\), the inverse of \(x + RF(x^n)\) is \(u + RF(x^n)\). In this way, for every \(\alpha \in \F_{q^n}\), we have
\[
\sigma^{-1}(\alpha) = ux\sigma^{-1}(\alpha) + v F(x^n) \sigma^{-1}(\alpha) = u \alpha x + vF(x^n)\sigma^{-1}(\alpha).
\]
This implies that 
\begin{equation}\label{inverso}
(\sigma^{-1}(\alpha)+RF(x^n)) (u + RF(x^n)) = (u + RF(x^n))(\alpha + RF(x^n)),
\end{equation}
making it consistent to denote \(u + RF(x^n)\) by \(x^{-1} + RF(x^n)\). As a consequence, $x^i + R F(x^n)$ is also a unit in $R_F$ for every $i \geq 1$. In the next, we denote by $x^{-i}+RF(x^n)$ the inverse of $x^i+RF(x^n)$.

\begin{lemma} \label{lm:inverseofx}
    There exists a polynomial $z(x^n) \in Z(R)$ , with $\deg(z(x^n))<sn$, such that $z(x^n)x^{ns-i}+RF(x^n)$ is the inverse of $x^i+RF(x^n)$, for every $i \in \{1,\ldots,ns\}$.
\end{lemma}

\begin{proof}
    Since \(x^{ns} + RF(x^n)\) is in the center of \(R_F\) so is its inverse. There exists an element $z(x^n)+RF(x^n) \in Z(R_F)$, with $\deg(z(x^n))<sn$, which is the inverse of $x^{ns}+RF(x^n)$. As a consequence, 
    \[
    z(x^n)x^{ns-i}x^i+RF(x^n)=1+RF(x^n),
    \]
    that proves the assertion.
\end{proof}

Our goal is to define a ring anti-isomorphism between $R_F$ and $R_{\hat{F}}$, where \(\hat{F}(y)\) is the {monic reciprocal polynomial} of $F(y)$, i.e.
\[
\hat{F}(y)=F_0^{-1}y^s F\left(\frac{1}{y}\right)=F_{0}^{-1}(1 + F_{s-1} y + \cdots + F_1 y^{s-1} + F_0 y^{s}).
\]
It is well known that $F(y) \in \F_q[y]$ is irreducible if and only if $\hat{F}(y)$ is irreducible. The main candidate for this mapping sends $x + R F(x^n)$ to the inverse of $x + R \hat{F}(x^n)$. To achieve this, we first recall a well known result that allows us to define a homomorphism between $R$ and a ring $S$, if we establish its action on $x$.

\begin{proposition} [see \textnormal{\cite[Proposition 2.4]{goodearl2004introduction}}] \label{prop:howtoextend}
Let $S$ be a ring, and assume that we have a ring homomorphism $\Phi:\F_{q^n} \rightarrow S$, and an element $x' \in S$ such that 
\begin{equation} \label{eq:relationtoextend}
    x'\Phi(\alpha)=\Phi(\sigma(\alpha))x',
\end{equation}
for every $\alpha \in \F_{q^n}$.
Then there is a unique ring homomorphism $\Psi: R \rightarrow S$ such that $\Psi |_{\fqn}=\Phi$ and $\Psi(x)=x'$. In particular, $\Psi$ is defined as
\begin{equation} \label{eq:extendedring}
\Psi: \sum_{i}a_ix^i \in R \longmapsto \sum_{i}\Phi(a_i)x'^i \in S
\end{equation}
\end{proposition}

For a ring $S$, we denote by $S^{op}$ the opposite ring of $S$. We first determine a correspondence between $R$ and $\left(R_{\hat{F}}\right)^{op}$.

\begin{lemma} \label{lm:RinRFopp}
The map 
    \[
    \Psi: \sum_ia_ix^i \in R \longmapsto \sum_i \sigma^{-i}(a_i)x^{-i} + R\hat{F}(x^n) \in \left(R_{\hat{F}}\right)^{op}
    \]
    is a surjective ring homomorphism from $R$ onto $\left(R_{\hat{F}}\right)^{op}$.
\end{lemma}

\begin{proof}
Let $\cdot$ denote the multiplication in $\left(R_{\hat{F}}\right)^{op}$. We get from \eqref{inverso} that, for any \(\alpha \in \mathbb{F}_{q^n}\), 
    \[
    \begin{split}
    (x^{-1}+R\hat{F}(x^n)) \cdot (\alpha + R\hat{F}(x^n)) &= (\alpha + R\hat{F}(x^n))(x^{-1}+R\hat{F}(x^n)) \\
    & = (x^{-1}+R\hat{F}(x^n))(\sigma(\alpha) + R\hat{F}(x^n))\\
    & = (\sigma(\alpha) + R\hat{F}(x^n)) \cdot (x^{-1}+R\hat{F}(x^n))
    \end{split}
    \]
    Thus, by taking $\Phi$ as the canonical inclusion map \(\mathbb{F}_{q^n} \to R_{\hat{F}}\), we get that the equation \eqref{eq:relationtoextend} is satisfied in $\left( R_{\hat{F}} \right)^{op}$ for $x'=x^{-1}+R\hat{F}(x^n)$. As a consequence, by Proposition \ref{prop:howtoextend}, we obtain that there exists a unique ring homomorphism $\Psi$ between $R$ and $\left( R_{\hat{F}} \right)^{op}$, defined by
    \[
    \Psi: \sum_ia_ix^i \in R \longmapsto \sum_i a_i \cdot (x^{-i} +R\hat{F}(x^n))=\sum_i \sigma^{-i}(a_i)x^{-i} + R\hat{F}(x^n) \in \left(R_{\hat{F}}\right)^{op}.
    \]
    Clearly, $\Psi$ is surjective, and the assertion follows.
\end{proof}

By using the above result, we are able to extend \cite[Lemma 26]{gomez2019dual} from the linear case to the current setting. 

\begin{theorem} \label{th:antiwithreciprocal}
   The map \begin{equation} \label{eq:definitionanti}
    \Theta: \sum_{i=0}^{ns-1}a_ix^i+RF(x^n) \in R_F \longmapsto \sum_{i=0}^{ns-1} \sigma^{-i}(a_i)x^{-i} + R\hat{F}(x^n) \in R_{\hat{F}}
    \end{equation}
    is an $E_F$-algebra anti-isomorphism between $R_F$ and $R_{\hat{F}}$.
\end{theorem}

\begin{proof}
    By using \Cref{lm:RinRFopp}, we obtain that the map
\[
    \Psi': \sum_ia_ix^i \in R \longmapsto \sum_i \sigma^{-i}(a_i)x^{-i} + R\hat{F}(x^n) \in R_{\hat{F}},
    \]
    is an anti-homomorphism of rings. We now compute the kernel of $\Psi'$, which is a twosided ideal of $R$. First, note that, since $F_0 \neq 0$, we have $\gcrd(\hat{F}(x^n), x) = 1$ in $R$. By \Cref{lm:inverseofx}, there exists an element $z(x^n) \in Z(R)$ , with $\deg(z(x^n))<sn$, such that $z(x^n)x^{ns}+R\hat{F}(x^n)$ is the identity in $R_{\hat{F}}$, and $z(x^n)x^{ns-i}+R\hat{F}(x^n)$ is the inverse of $x^i+R\hat{F}(x^n)$, for every $i \in \{1,\ldots,ns-1\}$. So, we have
    \[
    \begin{split}
         \Psi'(F(x^n))& =F_0+F_1x^{-n}+\cdots+F_{s-1}x^{-n(s-1)}+x^{-ns} +R\hat{F}(x^n) \\
         & = z(x^n)(F_0x^{ns}+F_1x^{n(s-1)}+\cdots+F_{s-1}x^{n}+F_s) +R\hat{F}(x^n) \\
         & = z(x^n) F_0 (F_0^{-1}(F_0x^{ns}+F_1x^{n(s-1)}+\cdots+F_{s-1}x^{n}+F_s)) +R\hat{F}(x^n) \\
         & = z(x^n) F_0 \hat{F}(x^n) +R\hat{F}(x^n) \\
         & =0+R\hat{F}(x^n).
    \end{split}
    \]
    Therefore, $\Psi'(F(x^n)) = 0 + R \hat{F}(x^n)$, implying that $R F(x^n)$ is contained in the kernel of $\Psi'$. By a standard degree argument, we obtain that $R F(x^n) = \ker(\Psi')$. Thus, $\Psi'$ induces the ring anti-isomorphism $\Theta$ between $R_F = R / R F(x^n)$ and $R_{\hat{F}}$ as defined in \eqref{eq:definitionanti}. Finally, it is easy to check that $\Theta$ is also an $E_F$-linear map, which proves our assertion.
\end{proof}

Next proposition describes the inverse of $\Theta$. 

\begin{proposition} \label{prop:inverseTheta}
Let $\Theta$ be as in \eqref{eq:definitionanti}. Then $\Theta^{-1}$ is the map
     \begin{equation} \label{eq:inversetheta}
     \sum_{i=0}^{ns-1}a_ix^i+R\hat{F}(x^n) \in R_{\hat{F}} \longmapsto \sum_{i=0}^{ns-1} \sigma^{-i}(a_i)x^{-i} + RF(x^n) \in R_{F}.
     \end{equation}
\end{proposition}

\begin{proof}
    Let $\overline{\Theta}$ denote the map defined as in \eqref{eq:inversetheta}, which is an \(\F_{q^s}\)-algebra anti-isomorphism by virtue of Theorem  \ref{th:antiwithreciprocal} applied to \(\hat{F}\). Since we know that \(\Theta\) is bijective, we only need to prove that $\overline{\Theta} \circ \Theta$ acts as the identity map to obtain $\overline{\Theta} = \Theta^{-1}$. Observe that this is an \(\F_{q^s}\)-algebra isomorphism, so we just need to show that it acts as the identity on a set of generators of the \(\F_{q^s}\)--algebra \(R_F\). If \(a \in \F_{q^n}\), then
    \[
\overline{\Theta}\Theta (a + RF(x^n)) = \overline{\Theta}(a + R\hat{F}(x^n)) =a + RF(x^n)
\]
and
\[
\begin{split}
\overline{\Theta}\Theta (x + RF(x^n)) &= \overline{\Theta}(x^{-1} + R\hat{F}(x^n)) = \overline{\Theta}(x + R\hat{F}(x^n))^{-1} \\
&= (x^{-1} + RF(x^n))^{-1} = x + RF(x^n), 
\end{split}
\]
which proves the assertion.
\end{proof}

With $T(y)=y-1$, a fundamental role in $\frac{R}{RT(x^n)} \cong M_n(\F_{q})$ is played by the \emph{adjoint} of a element, see \cite[pag. 480]{sheekey2016new}. Indeed, it provides the analogue of the transpose in $M_n(\F_q)$. More precisely, for an element $a=\sum_{i=0}^{n-1}a_ix^i+RT(x^n) \in R_T$, its \emph{adjoint} is defined to be the element \[\sum_{i=0}^{n-1}\sigma^{-i}(a_i)x^{-i}+RT(x^n)=\sum_{i=0}^{n-1}\sigma^{n-i}(a_i)x^{n-i}+RT(x^n) \in R_T\]

By using the anti-isomorphism $\Theta$ provided in \Cref{th:antiwithreciprocal}, we can extend this notion in the ring $R_F$.

\begin{definition} \label{def:adjoint}
  The \textbf{adjoint element} of  $a = \sum_{i=0}^{ns-1}a_ix^i + RF(x^n) \in R_F$ is 
    \[
    \Theta(a)=\sum_{i=0}^{ns-1} \sigma^{-i}(a_i)x^{-i} + R\hat{F}(x^{n}) \in R_{\hat{F}}.
    \]
\end{definition}

We observe that if $z(x^n) \in Z(R)$ is as in Lemma \ref{lm:inverseofx}, with $G(y)=\hat{F}(y)$, we have that 

\begin{equation} \label{eq:descryptionadjoint}
\Theta(a)=\sum_{i=0}^{ns-1} \sigma^{-i}(a_i)x^{-i} + R\hat{F}(x^n)=z(x^n)\sum_{i=0}^{ns-1} \sigma^{ns-i}(a_i)x^{ns-i} + R\hat{F}(x^n).
\end{equation}

Note that $R_F$ and $R_{\hat{F}}$ are both isomorphic to the matrix ring $M_n(\F_{q^s})$. We prove that the notion of adjoint given in Definition \ref{def:adjoint} is consistent with the usual notion of the adjoint of an element in $R_T \cong M_n(\F_q)$. 

Let \(\M_{R_F} : R_F \longrightarrow M_n(\F_{q^s})\)
be an $\F_{q^s}$-algebra isomorphism. We will show that the transpose of the matrix $\M_{R_F}(a) \in M_n(\F_{q^s})$ coincides with $\mathrm{M}_{R_{\hat{F}}}'(\Theta(a))$, for some $\F_{q^s}$-algebra isomorphism $\M_{R_{\hat{F}}}':R_{\hat{F}} \rightarrow M_n(\F_{q^s})$. To this aim, let 
\[\M_{R_{\hat{F}}}:R_{\hat{F}} \longrightarrow M_n(\F_{q^s})\]
be an $\F_{q^s}$-algebra isomorphism. We first observe that the anti-isomorphism $\Theta$ defined in Theorem \ref{th:antiwithreciprocal} allows us to define an $\F_{q^s}$-algebra anti-automorphism of $M_n(\F_{q^s})$: \[ \M_{R_{\hat{F}}} \Theta \M_{R_F}^{-1} : M_n(\F_{q^s}) \to M_n(\F_{q^s}) .\]

For a matrix \(A\), the notation \(A^\top\) stands for the transpose of \(A\).

\begin{theorem} \label{th:transposequotient}
There exists an $\F_{q^s}$-algebra isomorphism $\M_{R_{\hat{F}}}':R_{\hat{F}} \rightarrow M_n(\F_{q^s})$ such that 
\[
\M_{R_F}(a)^{\top}= \M_{R_{\hat{F}}}'(\Theta(a)),
\]
for all $a \in R_F$.
\end{theorem}

\begin{proof}
The map $\M_{R_{\hat{F}}} \Theta \M_{R_F}^{-1}$ is an anti-isomorphism of $M_n(\F_{q^s})$, so, as a consequence of Skolem-Noether Theorem, there exists a matrix $N \in \GL_n(\F_{q^s})$ such that 
\[\M_{R_{\hat{F}}}(\Theta(\M_{R_F}^{-1}(A)))=N A^{\top} N^{-1},\]
for every $A \in M_n(\F_{q^s})$. So, writing $A=\M_{R_F}(a)$, we get that
\[
N^{-1}\M_{R_{\hat{F}}}(\Theta(a))N=\M_{R_F}(a)^{\top},
\]
for every $a \in R_F$. Finally, by observing that
\[
M_{R_{\hat{F}}}': b \in R_{\hat{F}} \longmapsto N^{-1} M_{R_{\hat{F}}}(b) N \in M_n(\F_{q^s})
,\] 
is an $\F_{q^s}$-algebra isomorphism as well, we get the assertion.
\end{proof}

\section{Duality theory} \label{sec:duality}

The duality theory presented in this section is based, following the approach in \cite{gomez2020some}, on the notion of a Frobenius algebra.  A finite dimensional algebra \(A\) over a field \(K\) is said to be a \emph{Frobenius algebra} if there exists a non degenerate  bilinear form \(\langle -,-\rangle : A \times A \to K\) which is associative in the sense that \(\langle ab, c \rangle = \langle a, bc\rangle \) for all \(a,b,c \in A\). We say that such a bilinear form is a \emph{Frobenius bilinear form}. Alternatively, a Frobenius \(K\)--algebra may be defined by requiring that there is a linear form \(\varepsilon: A \to K\) whose kernel contains no nonzero right ideal. This linear form is known as a \emph{Frobenius functional} on \(A\). Frobenius bilinear forms and functionals are related by the equality \(\varepsilon (ab) = \langle a, b \rangle \), see e.g. \cite[Remark 2]{gomez2020some}.

For instance, any field \(K\) is a Frobenius algebra over every subfield \(k\), whenever the field extension \(K/k\) is finite. Any nonzero linear form \(\varepsilon : K \to k\) serves as a Frobenius functional. It is also well known that the full matrix ring \(M_n(K)\) is a Frobenius \(K\)--algebra with Frobenius functional \(\Tr : M_n(K) \to K \). From this, it is easily deduced that \(\varepsilon \Tr\) is a Frobenius functional for the \(k\)--algebra \(M_n(K)\).

From the foregoing discussion, and keeping the notation of the previous section, \(M_n(\F_{q^s})\) is a Frobenius algebra over \(\F_p\) with the Frobenius bilinear form 
\[
\langle -,- \rangle: M_n(\F_{q^s}) \times M_n(\F_{q^s}) \rightarrow \F_p 
\]
defined by 
\begin{equation} \label{eq:bilinearmatrix}
\langle A,B \rangle=\Tr_{q^s/p}(\Tr(AB)),
\end{equation}
for every $A,B \in M_n(\F_{q^s})$.

Another class of examples of Frobenius bilinear forms are defined on the \(\F_p\)--algebras \(R_F\) considered in previous sections. Given \(a, b \in R_F\), \((ab)_0\) stands for  term of degree \(0\) of the unique representative in \(R\) of degree less than \(ns\) of \(ab \in R_F\). The \(\F_p\)--algebra \(R_F\) is Frobenius according to the following theorem.

\begin{proposition} \label{th:bilinearformonRF}
    The \(\F_p\)--algebra \(R_F\) is Frobenius with bilinear Frobenius form 
    \[
    \langle -,- \rangle_F: R_F \times R_F \longrightarrow \F_p,\]
  defined by
    \begin{equation} \label{eq:bilinearoverRF}
    \langle a,b \rangle_F=\Tr_{q^n/p}\left((ab)_0\right)
     \end{equation}
\end{proposition}

\begin{proof}
Let \(\epsilon_F:R_F \to \F_p\) be the functional defined by \(\epsilon_F\left(\sum_{i=0}^{sn-1} g_i x^i\right) = \Tr_{q^n/p}(g_0)\), i.e. \(\langle a,b \rangle_F = \epsilon_F(ab)\). Hence \(\langle -,- \rangle_F\) is a Frobenius bilinear form if and only if \(\epsilon_F\) is a linear functional containing no nonzero left ideals. Linearity is clear. So let \(I \subseteq R_F\) be a left ideal such that \(\epsilon_F(I) = 0\). If \(I \neq 0\), then \(I = Rg/RF(x^n)\) for some proper left divisor \(g\) of \(F(x^n)\). Since \(\epsilon_F(g) = 0\), it follows \(g_0 = 0\). Therefore \(x\) left divides \(F(x^n)\) and \(F_0 = 0\) a contradiction. Consequently \(I = 0\) and \(\epsilon_F\) is a Frobenius functional. 
\end{proof}

Next, we relate the bilinear form defined as in \eqref{eq:bilinearmatrix} over $M_n(\F_{q^s})$ and the bilinear form as in \eqref{eq:bilinearoverRF} defined over $R_F$. 

\begin{theorem} \label{th:correspondencebilinear}
Let $\langle -,- \rangle$ be the bilinear form defined as in \eqref{eq:bilinearmatrix} over $M_n(\F_{q^s})$ and let $\langle -, -\rangle_F$ defined as in \eqref{eq:bilinearoverRF} over $R_F$.
Then there exists an invertible element $U \in \GL_n(\F_{q^s})$ such that 
\[
\langle M_{R_F}(a), M_{R_F}(b) U  \rangle= \langle a,b \rangle_F,
\]
for every $a,b \in R_F$.
\end{theorem}

\begin{proof}
Since \(\M_{R_F} : R_F \rightarrow M_n(\F_{q^s}) \) is an \(\F_p\)--algebra isomorphism, we get that \([a,b] = \langle M_{R_F}(a), M_{R_F}(b)\rangle \) is a Frobenius bilinear form on \(R_F\). By \cite[Th. 3.1]{jans1959frobenius}, there is a unit \(u \in R_F\) such that
\(\langle a, b\rangle_F = [a,bu]\) for all \(a,b \in R_F\). Setting \(U = M_{R_F}(u)\) gives the desired equality. 
\end{proof}

Recall that, given a non degenerate  bilinear form \(\langle -,- \rangle\) on a finite dimensional vector space over a field \(K\), the map \(V \to V^*\) given by the assignment \(v \mapsto \langle -,v\rangle\) is an isomorphism of vector spaces. Here, \(V^*\) denotes vector space of all linear forms defined on \(V\). Given any vector subspace \(W\) of \(V\), we have an injective linear map
\[
\left( V/W\right)^* \to V^* \cong V,
\]
whose image is \[W^\perp = \{v \in V : \langle w,v \rangle = 0 \; \forall\;  w \in W\}.\] As a consequence, we get the well known dimension formula
\begin{equation} \label{eq:dimensionalformula}
\dim_K V = \dim_K W + \dim_K W^\perp.
\end{equation}

\section{Application on rank-metric codes}

The adjoint and duality theory for quotients of skew polynomial rings developed in sections \ref{sec:adjoint} and \ref{sec:duality}, allows us to extend the study of the recently introduced families of MRD codes from \cite{sheekey2020new} and \cite{lobillo2025quotients}. Specifically, we determine the adjoint and dual codes of these families. Additionally, we compute the \emph{idealisers, centralisers} and the \emph{centre} of the codes contained in these families for choices of the parameters that have not yet been addressed in the existing literature. These computations prove that these two families provide new examples of MRD codes for an extended set of infinite parameters. \\

We begin by recalling the essential notions and key results related to rank-metric codes relevant to our work. Let $\F$ be a finite field. A \emph{rank-metric code} is a subset $\C$ of the matrix space $M_{m \times n}(\F)$ endowed with the rank-distance metric:
\[
d(A,B) = \rk(A-B).
\]
The minimum distance $d(\C)$ of a code $\C $ is given by
\[
d(\C)=\min\{\rk(A-B) \colon A,B\in \C, A\ne B\}.
\]
For a subfield $\F'\leq \F$, a code $\C$ is said \emph{$\F'$-linear} if it is an $\F'$-subspace of $M_{m\times n}(\F)$. When $\F'$ is the prime subfield of $\F$, the code $\C$ is called additive. 

Any rank-metric code $\C$ of $M_{m\times n}(\F)$ satisfy the \emph{Singleton-like bound} \cite{delsarte1978bilinear}. Precisely, if $\C$ has a minimum distance $d$, then
\begin{equation} \label{eq:singleton}
|\C|\leq |\F|^{\max\{m,n\}(\min\{m,n\}-d+1)}.
\end{equation}
A code attaining this bound is known as a \emph{Maximum Rank Distance (MRD) code}.

In what follows, we will focus on the case $n=m$. Starting from a code $\C$, it is possible to define two further codes. 

\begin{definition}
    Let $\C$ be a rank-metric code in $M_n(\F)$. the \emph{adjoint code} of $\C$ is
\[ \C^\top =\{X^{\top} \colon X \in \C\} \subseteq M_n(\F). \]
 The \emph{dual code} of a rank-metric code $\C$ is
\[ \C^\perp = \{ Y \in M_n(\F) \colon \langle X,Y \rangle_{\rk}=0, \, \text{for all } \, X \in \C \} \subseteq M_n(\F), \]
where $\langle -,- \rangle_{\rk}$ denotes the bilinear form on $M_n(\F)$ defined by
\begin{equation} \label{eq:bilinearrank}
\langle X,Y \rangle_{\rk}= \mathrm{Tr}_{\F/\F'}\left(\mathrm{Tr}(XY^{\top})\right), 
\end{equation} 
where $\F'$ is the prime subfield of $\F$.
\end{definition}

Clearly, the adjoint of an MRD code is an MRD code, as well, and by using association schemes, Delsarte in \cite{delsarte1978bilinear} proves the dual of an MRD code is an MRD code. 

To distinguish rank-metric codes, we recall the notion of equivalence. For an automorphism $\rho$ of $\F$ and a matrix $A \in M_n(\F)$, by $A^{\rho}$ we denote the matrix obtained by applying $\rho$ to all its entries.

\begin{definition}
Two rank-metric codes $\C,\C' \subseteq M_n(\F)$ are equivalent if
\begin{equation} \label{eq:equivalencecode}
\C'=U \C^{\rho} V=\{UA^{\rho}V \colon A \in \C\},
\end{equation}
where $U,V \in \GL_n(\F)$, and $\rho$ is an automorphism of $\F$.
\end{definition}

As before, we assume that $F(y) \in \F_q[y]$ is a monic irreducible polynomial of degree $s$, with nonzero constant coefficient $F_0$.  According to the previous sections,  $R_F = R / RF(x^n)$ and $M_n(\F_{q^s})$ are isomorphic  $\F_{q^s}$-algebras via some isomorphism $\M_{R_F}$. Therefore, for any subset $C$ of $R_F$, we can consider its image $\M_{R_F}(C)$ in $M_n(\F_{q^s})$, which turns out to be a rank-metric code. 

We first prove that if a subset of \(R_F\) is represented using different $\F_{q^s}$-algebra isomorphisms, then the resulting rank-metric codes in $M_n(\F_{q^s})$ are equivalent.

\begin{lemma} \label{lm:indipendencerepresentation}
    Let $\M_{R_F},\M'_{R_F}:R_F \rightarrow M_n(\F_{q^s})$ be $\F_{q^s}$-algebra isomorphisms. And let $C$ be a subset of $R_F$. Then the rank-metric codes $\C_1=\M_{R_F}(C)$ and $\C_2=\M'_{R_F}(C)$ in $M_n(\F_{q^s})$ are equivalent. 
\end{lemma}

\begin{proof}
By Skolem-Noether's Theorem, there exists $N \in \GL_n(\F_{q^s})$ such that
\[
\M_{R_F}'(a)=N\M_{R_F}(a) N^{-1}
\]
for any $a \in R_F$. As a consequence,
\[
\C_2=N \C_1 N^{-1},
\]
that proves the assertion. 
\end{proof}

Therefore, the representation of the rank-metric code in $M_n(\F_{q^s})$, as far as its equivalence class concerns, does not depend on the choice of the isomorphism between $R_F$ and $M_n(\F_{q^s})$.

Delsarte \cite{delsarte1978bilinear}, and later Gabidulin \cite{gabidulin1985theory}, proved the existence of MRD codes over every finite field and for all parameters. More precisely, they constructed $\F_q$-linear MRD codes in $M_n(\F_q)$ with size $q^{nk}$ and minimum distance $n-k+1$, for any $1< k<n$. These codes are often known as \emph{Gabidulin codes}. Later, the family of Gabidulin codes was extended by Sheekey to the family of \emph{twisted Gabidulin codes} and then by Lunardon, Trombetti and Zhou in \cite{lunardon2018generalized}. These families provide the same set of parameters as Gabidulin codes, but they are inequivalent to them (cf. \cite[Theorem 7]{sheekey2016new}). Another relevant family of MRD codes is defined by the Trombetti-Zhou codes \cite{trombetti2018new}, that are $\F_q$-linear MRD codes in $M_n(\F_q)$,
but requiring $q$ odd and $n$ even. In 2020, Sheekey's groundbreaking work \cite{sheekey2020new} introduced a large family of MRD codes by quotients of skew polynomials $R_F$. These codes include both Gabidulin and twisted Gabidulin codes. Let us record this result for later reference.

\begin{theorem} [see \textnormal{\cite[Theorem 7]{sheekey2020new}}]
Let $\rho \in \Aut(\fqn)$ and let $\K=\mathrm{Fix}(\rho) \cap \fq$. Let $1\leq k < n$ be a positive integer. Then the set 
\begin{equation} \label{eq:johncodes}
S_{n,s,k}(\eta,\rho,F)=\left\{a_0+\sum_{i=1}^{sk-1}a_ix^i+\eta \rho(a_0)x^{ks} +RF(x^n) \colon a_i \in \fqn  \right\} \subseteq R_F,
\end{equation}
defines a $\K$-linear MRD code \(\mathcal{C}\) in $M_n(\F_{q^s})$ with $\dim_{\K}(\C)=[\F_{q^n}:\K]sk$ and having minimum distance $n-k+1$, for any $\eta \in \fqn$ such that $\N_{\fqn/\K}(\eta) \N_{\fq/\K}((-1)^{sk(n-1)}F_0^k) \neq 1$.
\end{theorem}

Similarly, by using the quotients of skew polynomials $R_F$, in \cite{lobillo2025quotients} a new large family of MRD codes has been constructed that properly contains the \emph{Trombetti-Zhou codes} \cite{trombetti2018new}. This family is defined according to the following theorem. 

\begin{theorem} [see \textnormal{\cite[Theorem 6.1.]{lobillo2025quotients}}]
Assume that $q$ is an odd prime power. Let $n=2t \geq 2$. For a positive integer $1\leq k<n$, the set
\begin{equation} \label{eq:finiteextensiontrombzhou}
D_{n,s,k}(\gamma,F)=\left\{ a_0'+\sum_{i=1}^{sk-1} a_i x^i + \gamma a_0'' x^{sk} +RF(x^n) \colon a_i \in \F_{q^n}, a_0',a_0'' \in \F_{q^t} \right\} \subseteq R_F,
\end{equation}
defines an $\F_q$-linear MRD code \(\mathcal{C}\) in $ M_{n}(\F_{q^s})$ with $\dim_{\F_q}(\C)=nsk$ and minimum distance $n-k+1$ for any $\gamma \in \F_{q^n}$ such that $(-1)^{ks}F_0^{k}\N_{\F_{q^n}/\F_q}(\gamma)$ is not a square in $\F_q$.
\end{theorem}

\begin{remark} \label{rK:MRDs1}
In the case where $s=1$ and $F(y) = y - 1$, the codes $S_{n,1,k}(\eta, \rho, F)$ correspond to (generalized) Gabidulin codes \cite{delsarte1978bilinear,gabidulin1985theory,Gabidulins} or twisted Gabidulin codes \cite{sheekey2016new,otal2016additive,lunardon2018generalized}, depending on whether $\eta = 0$ or not, respectively. Meanwhile, the codes $D_{n,1,k}(\gamma, F)$ are exactly the Trombetti-Zhou codes \cite{trombetti2018new}.
\end{remark}

\begin{remark}
It is worth noting that quotients of the skew polynomial ring \( R_F \) have also been studied in the context of cyclic Galois extensions \( \mathbb{L}/\mathbb{K} \), leading to new nonassociative division algebras and MRD codes over matrix spaces \( M_n(\mathbb{D}) \), where \( \mathbb{D} \) is a (non necessarily associative) division algebra; cf. \cite{sheekey2020new,thompson2023division,lobillo2025quotients}.   
\end{remark}

Other few families of MRD codes are known in the literature, but only for specific parameters. We summarize in Table \ref{tab:parameters} the known MRD codes with their respective references. We will not consider MRD codes in $M_n(\F)$ with minimum distance $n$, as they correspond to semifields and are beyond the scope of this paper.

The problem of determining the adjoint and dual codes of the families $S_{n,s,k}(\eta,\rho,F)$ and $D_{n,s,k}(\gamma,F)$ has not been addressed in literature. By making use of the tools developed in Sections \ref{sec:adjoint} and \ref{sec:duality}, we are able to solve this problem. Let us start by determining the adjoint codes of the families $S_{n,s,k}(\eta,\rho,F)$ and $D_{n,s,k}(\gamma,F)$.

\begin{proposition} \label{prop:determinantionadjoint}
Let $\hat{F}(y)$ be the monic reciprocal polynomial of the irreducible polynomial $F(y) \in \F_q[y]$. For any $1 \leq k<n$, the following hold. 

\begin{enumerate}[I)]
\item\label{adjointS} The adjoint code of $S_{n,s,k}(\eta,\rho,F) \subseteq R_F \cong M_n(\F_{q^s})$ is equivalent to
\[
S_{n,s,k}(\rho^{-1}(\eta^{-1}),\rho^{-1}\circ \sigma^{ks},\hat{F}) \subseteq R_{\hat{F}} \cong M_n(\F_{q^s}).
\]
\item\label{adjointD} The adjoint code of  $D_{n,s,k}(\gamma,F) \subseteq R_F \cong M_n(\F_{q^s})$ is equivalent to
 \[D_{n,s,k}\left(\sigma^{s(n-k)}\left(\frac{1}{\gamma}\right),\hat{F}\right)\subseteq R_{\hat{F}} \cong M_n(\F_{q^s}).
 \]
\end{enumerate}
\end{proposition}

\begin{proof}
As proved in \Cref{lm:indipendencerepresentation}, the image of a subset $\C$ of $R_F$ under different $\F_{q^s}$-algebra isomorphisms $R_F \cong M_n(\F_{q^s})$ gives equivalent codes.  So, let fix $M_{R_F}:R_F \rightarrow M_n(\F_{q^s})$ be an $\F_{q^s}$-algebra isomorphism. 

\eqref{adjointS} Let 
\[
\C=\{M_{R_F}(a) \colon a \in S_{n,s,k}(\eta,\rho,F)\} \subseteq M_n(\F_{q^s}),
\]
we need to determine $\mathcal{C}^{\top}$. By \Cref{th:transposequotient}, we know that there exists an $\F_{q^s}$-algebra isomorphism $\M_{R_{\hat{F}}}:R_{\hat{F}}\rightarrow M_n(\F_{q^s})$ such that \[
\M_{R_F}(a)^{\top}=\M_{R_{\hat{F}}}(\Theta(a)),
\]
for any $a \in R_F$. Let $z(x^n) \in Z(R)$ be as in \Cref{lm:inverseofx}, with $G(y)=\hat{F}(y)$. We need to determine $\Theta(a)$, for any $a = a_0+\sum_{i=0}^{sk-1}a_ix^i+\eta \rho(a_0)x^{ks} +RF(x^n) \in S_{n,s,k}(\eta,\rho,F)$. By \eqref{eq:descryptionadjoint}, 
\[
\begin{split}
\Theta(a) & = \Theta\left( a_0+\sum\limits_{i=1}^{sk-1}a_ix^i+\eta \rho(a_0)x^{ks}  + RF(x^n)\right) \\ & = z(x^n)\left( a_0x^{ns}+\sum\limits_{i=1}^{sk-1}\sigma^{sn-i}(a_i)x^{sn-i}+\sigma^{s(n-k)}(\eta \rho(a_0))x^{s(n-k)} \right) +R\hat{F}(x^n) \\
& = z(x^n)\left( a_0x^{sk}+\sum\limits_{i=1}^{sk-1}\sigma^{sn-i}(a_i)x^{sk-i}+\sigma^{s(n-k)}(\eta \rho(a_0)) \right)x^{s(n-k)} +R\hat{F}(x^n)
\end{split}
\]
Observe that 
\[
a_0x^{sk}+\sum\limits_{i=1}^{sk-1}\sigma^{sn-i}(a_i)x^{sk-i}+\sigma^{s(n-k)}(\eta \rho(a_0))  + R\hat{F}(x^n) \in S_{n,s,k}(\rho^{-1}(\eta^{-1}),\rho^{-1}\circ \sigma^{ks},\hat{F}),
\] 
and set \(M = M_{\hat{F}}(z(x^n)+ R\hat{F}(x^n))\), 
\(N = M_{\hat{F}}(x^{n(s-k)}+R\hat{F}(x^n))\), which are invertible matrices.  
We have shown so far that
\[
\begin{split}
\C^{\top} &= \{M_{R_F}(a)^{\top} \colon a \in S_{n,s,k}(\eta,\rho,F)\} \\
&= \{M_{R_{\hat{F}}}(\Theta(a)) \colon a \in S_{n,s,k}(\eta,\rho,F)\} \\
&\subseteq \{M M_{\hat{F}}(b) N : b \in S_{n,s,k}(\rho^{-1}(\eta^{-1}),\rho^{-1}\circ \sigma^{ks},\hat{F})\}
\end{split}
\]
The last inclusion is an equality since both sets are vector spaces of the same dimension over the field \(\mathbb{K} = \mathrm{Fix}(\rho) \cap \mathbb{F}_q = \mathrm{Fix}(\rho^{-1} \sigma^{sk}) \cap \mathbb{F}_q\). Hence, \(\mathcal{C}^{\top}\) is equivalent to  \(S_{n,s,k}(\rho^{-1}(\eta^{-1}),\rho^{-1}\circ \sigma^{ks},\hat{F})\).

\eqref{adjointD} If \(a = a_0'+\sum_{i=1}^{sk-1} a_i x^i + \gamma a_0'' x^{sk} +RF(x^n) \in D_{n,s,k}(\gamma,F)\),  then, analogously to the computation of part \eqref{adjointS}, we get 


\[
\begin{split}
\Theta (a) & = z(x^n)\left(a_0'x^{ns}+ \sum\limits_{i=1}^{sk-1}\sigma^{sn-i}(a_i)x^{ns-i} +\sigma^{s(n-k)}(\gamma a_0'')\right)x^{s(n-k)} + R\hat{F}(x^n) \\
& =z(x^n)\sigma^{s(n-k)}(\gamma) \left(\sigma^{s(n-k)}\left(\frac{1}{\gamma}\right)a_0'x^{ns}+ \sum\limits_{i=1}^{sk-1}\sigma^{sn-i}(a_i)x^{ns-i} +\sigma^{s(n-k)}(a_0'')\right)x^{s(n-k)} \\
&\quad + R\hat{F}(x^n). 
\end{split}
\]
Now, proceed as in part \eqref{adjointS}. 
\end{proof}

Now, we deal with the dual codes of the rank-metric codes in the skew polynomial framework $R_F \cong M_n(\F_{q^s})$. The theory of duality for rank-metric codes is built on the bilinear form $\langle -,-\rangle_{\rk}$ defined as in \eqref{eq:bilinearrank}, see \cite[\S 3]{delsarte1978bilinear} and \cite[\S 1.5]{sheekey2016new}. On the other hand, recall that on $M_n(\F_{q^s})$ we have considered the Frobenius bilinear form $\langle -,- \rangle$ as defined in \eqref{eq:bilinearmatrix}, i.e., $\langle A,B \rangle = \Tr_{q^s/p}(\Tr(AB))$ for every $A,B \in M_n(\F_{q^s})$. However, we note that when we work with square matrices, the induced theories of duality are related by a transposition. More precisely, the following relation holds:
\[
\langle A,B \rangle_{\rk}=\langle A,B^{\top} \rangle.
\]
As a consequence, for a subset $\C$ of $M_n(\F_{q^s})$, the dual 
\[
\C^{\perp'}=\{B \colon \langle A, B \rangle=0, \mbox{ for every } A \in \C\}\]
with respect to the bilinear form $\langle -,- \rangle$ and the dual with respect to the bilinear form $\langle -,- \rangle_{\rk}$ are related by 
\begin{equation} \label{eq:relationduals}
(\C^{\perp'})^{\top} = \C^{\perp}.\end{equation} Thus, to determine the dual of a rank-metric code, we just need to compute the adjoint of the dual code of $\C$ with respect to the Frobenius bilinear form $\langle -,- \rangle$.

We are so ready to determine the dual of the codes in the families $S_{n,s,k}(\eta,\rho,F)$ and $D_{n,s,k}(\gamma,F)$. 

\begin{proposition} \label{prop:determinantiondual}
 Let $\hat{F}(y)$ be the monic reciprocal polynomial of $F(y)$. For any $1 \leq k<n$, the following hold. 

\begin{enumerate}[I)]
    \item\label{dualS} The dual code of $S_{n,s,k}(\eta,\rho,F)$ in $R_F \cong M_n(\F_{q^s})$ is equivalent to
 \[
 S_{n,s,n-k}(\rho^{-1}(\eta F_0), \rho^{-1},\hat{F}) \subseteq R_{\hat{F}} \cong M_n(\F_{q^s}).
 \]
 \item\label{dualD} The dual code of  $D_{n,s,k}(\gamma,F)$  in $R_F \cong M_n(\F_{q^s})$ is equivalent to
 \[D_{n,s,n-k}(\sigma^{sk}(\gamma),\hat{F}) \subseteq R_{\hat{F}} \cong M_n(\F_{q^s}).
 \]
 \end{enumerate}
\end{proposition}

\begin{proof}
As in the proof of \Cref{prop:determinantionadjoint}, we fix $M_{R_F}:R_F \rightarrow M_n(\F_{q^s})$ to be an $\F_{q^s}$-algebra isomorphism. Let $\langle -, - \rangle_F$ be the bilinear form defined as in \eqref{eq:bilinearoverRF} over $R_F$. \\
\emph{\ref{dualS})}  Let $S=S_{n,s,k}(\eta,\rho,F)$, we start by computing the dual $S^{\perp}$ of $S$ with respect to the bilinear form $\langle -,- \rangle_F$ of $R_F$, i.e. 
\[
S^{\perp}=\{b \in R_F \colon \langle a,b \rangle_F=0, \mbox{ for every }a \in S\}.
\]
Clearly, every monomial $\alpha x^i+RF(x^n)$, with $i \in \{1,\ldots,s(n-k)-1\}$ is orthogonal to the elements of $S$. Moreover, for any $\alpha \in \F_{q^n}$, we have that \[c = \rho^{-1}(\eta F_0\sigma^{sk}(\alpha))+\alpha x^{s(n-k)}+RF(x^n)\] is orthogonal to any element of $S$. Indeed, if $a=\sum_{i=0}^{sk-1}a_ix^i+\eta \rho(a_0)x^{sk} + RF(x^n)$ we have 
\[
\begin{split}
\langle a, c \rangle_F &= \Tr_{q^n/p}\left( \rho^{-1}(\eta F_0\sigma^{sk}(\alpha))a_0-\sigma^{sk}(\alpha) \eta \rho(a_0)F_0  \right) \\
&= \Tr_{q^n/p}\left( \rho^{-1}(\eta F_0\sigma^{sk}(\alpha))a_0)-\Tr_{q^n/p}(\sigma^{sk}(\alpha) \eta \rho(a_0)F_0  \right) \\ 
&= \Tr_{q^n/p}\left(\rho( \rho^{-1}(\eta F_0\sigma^{sk}(\alpha))a_0))-\Tr_{q^n/p}(\sigma^{sk}(\alpha) \eta \rho(a_0)F_0  \right)  \\
&= \Tr_{q^n/p}\left( \eta F_0\sigma^{sk}(\alpha)\rho(a_0)\right) - \Tr_{q^n/p}\left(\sigma^{sk}(\alpha)  \eta \rho(a_0)  F_0 \right) \\
&= 0
\end{split}
\]
As a consequence, \[
\begin{split}
S' &= \left\{\rho^{-1}(\eta F_0)\rho^{-1}(\sigma^{sk}(\alpha))+\sum\limits_{i=1}^{s(n-k)-1}a_ix^i+\alpha x^{s(n-k)}+RF(x^n), a_i \in \F_{q^s}\right\}\\
&= \left\{b+\sum\limits_{i=1}^{s(n-k)-1}a_ix^i+\sigma^{-sk}(\eta^{-1}F_0^{-1}) \sigma^{-sk}(\rho(b)) x^{s(n-k)}+RF(x^n), a_i \in \F_{q^s}\right\} \\
&= S_{n,s,n-k}(\sigma^{-sk}(\eta^{-1}F_0^{-1}), \sigma^{-sk} \circ \rho,F)
\end{split}
\]
is contained in $S^{\perp}$. Since $\langle -,- \rangle_F$ is a bilinear non degenerate  form, by \eqref{eq:dimensionalformula}, we have  \[
\dim_{\F_p}(S^{\perp})=n^2se-\dim_{\F_p}(S^{\perp})
\]
So, we get
\[
\lvert S' \rvert= \lvert S_{n,s,n-k}(\sigma^{-sk}(\eta^{-1}F_0^{-1}), \sigma^{-sk} \circ \rho,F) \rvert=q^{ns(n-k)}=\lvert S^{\perp} \rvert
\]
Thus, $S'=S^{\perp}$.

Now, let 
\(
\C=\{M_{R_F}(a) \colon a \in S\} \subseteq M_n(\F_{q^s}).
\)
We start by determine $\mathcal{C}^{\perp'}$, i.e. the dual of $\C$ with respect of $\langle -, - \rangle$. Then, by computing the adjoint of $\mathcal{C}^{\perp'}$ and by the relation in \eqref{eq:relationduals}, we will obtain $\C^{\perp}$. We know, by \Cref{th:correspondencebilinear}, there exists an invertible element $U \in \GL_n(\F_{q^s})$ such that 
\[
\langle M_{R_F}(a), M_{R_F}(b) U  \rangle= \langle a,b \rangle_F,
\]
for all \(a,b \in R_F\). As a consequence, we have
\[
\begin{split}
\C^{\perp'} &= \{\M_{R_F}(b)U: \langle a,b   \rangle_F=0, \mbox{ for every }a \in R_F\} \\
&= \{\M_{R_F}(b): \langle a,b  \rangle_F=0, \mbox{ for every } a \in R_F\}U \\
&= \{\M_{R_F}(b): b \in S_{n,s,n-k}(\sigma^{-sk}(\eta^{-1}F_0), \sigma^{-sk} \circ \rho,F) \}U
\end{split}
\]

Therefore, $\C^{\perp'}$ is equivalent to $S_{n,s,n-k}(\sigma^{-sk}(\eta^{-1}F_0^{-1}), \sigma^{-sk} \circ \rho,F)$ in $R_F \cong M_n(\F_{q^s})$. Finally, $\C^{\perp}=(\C^{\perp'})^{\top}$ is equivalent to to the adjoint of $S_{n,s,n-k}(\sigma^{-sk}(\eta^{-1}F_0), \sigma^{-sk} \circ \rho,F)$ that is
\[
S_{n,s,n-k}(\rho^{-1}(\eta F_0),  \rho^{-1},\hat{F}) \subseteq R_{\hat{F}} \cong M_n(\F_{q^s})
\]
by \emph{\ref{adjointS})} of \Cref{prop:determinantionadjoint}, that proves our assertion. \\
\emph{\ref{dualD})} We argue as in \emph{\ref{dualS})}. Let $D=D_{n,s,k}(\gamma,F)$. We start by computing the dual $D^{\perp}$ of $D$ with respect to the bilinear form $\langle -,- \rangle_F$ of $R_F$. 
Clearly, every monomial $\alpha x^i+RF(x^n)$, with $i \in \{1,\ldots,s(n-k)-1\}$ are orthogonal to the elements of $D$. Let now $\zeta \in \F_{q^n}^*$ be a nonzero element such that $\Tr_{q^n/q^{n/2}}(\zeta \gamma)=0$. It is easy to check that for any $\alpha,\beta \in \F_{q^{n/2}}$, the element \[\alpha \gamma \zeta +\alpha \zeta x^{s(n-k)}+RF(x^n)\] is orthogonal to any element of $D$.
As a consequence, the set \[
\begin{split}
&\left\{a_0'\gamma \zeta+\sum\limits_{i=1}^{s(n-k)-1}a_ix^i+a_0''\zeta x^{s(n-k)}+RF(x^n) \colon a_i \in \F_{q^n},a_0',a_0'' \in \F_{q^{n/2}} \right\} \\
&= \zeta \gamma \left\{a_0'+\sum\limits_{i=1}^{s(n-k)-1}a_ix^i+a_0''\frac{1}{\gamma} x^{s(n-k)} + RF(x^n) \colon a_i \in \F_{q^n},a_0',a_0'' \in \F_{q^{n/2}} \right\} \\
&= \zeta \gamma D_{n,s,n-k}(1/\gamma,F)
\end{split}
\]
is contained in $D^{\perp}$ and by a dimensional argument we have that it coincides with $D^{\perp}$.

Now, let  \(
\C=\{M_{R_F}(a) \colon a \in D\} \subseteq M_n(\F_{q^s}). \)
The former computation shows that $\C^{\perp'}$ is equivalent to $D_{n,s,n-k}(1/\gamma,F)$ in $R_F \cong M_n(\F_{q^s})$. Finally, $\C^{\perp}=(\C^{\perp'})^{\top}$ is equivalent to to the adjoint of $D_{n,s,n-k}(1/\gamma,F)$ that is
\[
D_{n,s,n-k}(\sigma^{sk}(\gamma),\hat{F}) \subseteq R_{\hat{F}} \cong M_n(\F_{q^s})
\]
by \emph{\ref{adjointD})} of \Cref{prop:determinantionadjoint}, that proves our assertion.
\end{proof}

\begin{remark}
As noted in \Cref{rK:MRDs1}, the family $S_{n,1,k}(\gamma,\rho,F)$ includes both Gabidulin and twisted Gabidulin codes. Note that if $F(y) = y - 1$, we have $\hat{F}(y) = F(y) = y - 1$. Thus, \ref{adjointS}) of \Cref{prop:determinantionadjoint} and \ref{adjointS}) of \Cref{prop:determinantiondual} also includes the calculation of the adjoint and dual codes of Gabidulin codes and twisted Gabidulin codes, cf. \cite[Theorem 6]{sheekey2016new}. Similarly, the family $D_{n,1,k}(\eta,F)$ corresponds to Trombetti-Zhou codes. In this case, \ref{adjointD}) of \Cref{prop:determinantionadjoint} and \ref{dualD}) of \Cref{prop:determinantiondual} includes the determination of the adjoint and dual codes of Trombetti-Zhou codes, cf. \cite[Propositions 4 and 5]{trombetti2018new}.
\end{remark}

In \cite{sheekey2020new} and in \cite{lobillo2025quotients}, it is proved that the families $S_{n,s,k}(\eta,\rho,F)$ and $D_{n,s,k}(\gamma,F)$ contain new MRD codes for infinitely many choices of the parameters $s$ and $n$, when $k \leq  n/2$, cf. \cite[Theorem 11]{sheekey2020new} and \cite[Theorem 6.3]{lobillo2025quotients}.  As a result, the families $S_{n,s,k}(\eta,\rho,F)$ and $D_{n,s,k}(\gamma,F)$ represent the largest known families of MRD codes. Thanks to the tools developed here, we are able to extend this result to the case $k > n/2$.

First, we recall the notion of idealisers, centralisers and centre of a rank-metric code. These are algebraic constructions with precedents in the study of noncommutative rings, which generalize for instance the notions of the nuclei and centre of a (non-necessarily associative) division algebras. They are developed in the realm of Coding Theory in  \cite{liebhold2016automorphism,lunardon2018nuclei,sheekey2020new}. In what follows, all codes are additive.

\begin{definition} Let $\C$ be a rank-metric code in $M_n(\F)$, with $\F$ a finite field. 
The \textbf{left idealiser} $\lid(\C)$ and the \textbf{right idealiser} $\rid(\C)$ are defined as
\[
\lid(\C) = \{A \in M_n(\F) \colon  A\C\subseteq \C\}
\]
and
\[
\rid(\C) = \{B \in M_n(\F) \colon \C B  \subseteq\C\},
\]
respectively. \\
The \textbf{centraliser} $\cen(\C)$ is defined as
\[
\cen(\C) = \{A \in M_{n}(\F) \colon AX=XA \mbox{ 
for every }X\in \C\}.
\]
The \textbf{centre} $Z(\C)$ of $\C$ is defined as the intersection of the left idealiser and the centraliser.
\[
Z(\C) = \lid(\C)\cap \cen(\C).
\]
\end{definition}

These objects are subrings of $M_n(\F)$. For an MRD code $\C$, its left idealiser $\lid(\C)$ and right idealiser $\rid(\C)$ turn out to be fields (see \cite[Corollary 5.6]{lunardon2018nuclei}).
We prove that for any MRD code $\C$, its centraliser - and hence its centre - is also a field. Moreover, if the minimum distance of $\C$ is not equal to $n$, then the centraliser is isomorphic to $\F$.

\begin{proposition} \label{prop:centraliserisfield}
Let $\C$ be a rank-metric code in $M_n(\F)$. Then, the centraliser $\cen(\C)$ of $\C$ contains a field isomorphic to $\F$. Moreover, if $\C$ is an MRD code, then $\cen(\C)$ is a field.
\end{proposition}

\begin{proof}
The first claim is clear because the center of \(M_n(\F)\) is \(\F\). 
Now assume that $\C$ is an MRD code. Suppose, for contradiction, that there exists a nonzero matrix $A \in \cen(\C)$ that is not invertible. By definition of the centraliser, we have
\begin{equation} \label{eq:conditioncentraliser}
AX = XA,
\end{equation}
for every $X \in \C$. Since $A$ is not invertible, there exists a vector $v \in \F^n$ such that $vA = 0$. Assume that the $i$-th row of $A$ is nonzero, and let $w \in \F^n$ be the standard unit vector with a $1$ in the $i$-th position and $0$ elsewhere. Because $\C$ is an MRD code, by \cite[Theorem 5.1]{lunardon2018nuclei}, there exists a codeword $Y \in \C$ such that $vY = w$. Using equation~\eqref{eq:conditioncentraliser} with $X = Y$ and multiplying both sides on the left by $v$, we get
\[
vAY = vYA.
\]
Since $vA = 0$, the left-hand side is zero, so $0 = vYA = wA$. But $wA$ is the $i$-th row of $A$, which is nonzero by assumption, yielding a contradiction. Hence, $\cen(\C)$ is  a finite division ring and, by Wedderburn's Theorem, a field.
\end{proof}

\begin{theorem} \label{cor:centraliserisfieldMRD}
Let $\C$ be an MRD code in $M_n(\F)$, with $\F$ a finite field. If $d(\C) < n$, then
\[
\cen(\C) \cong \F.
\]
\end{theorem}

\begin{proof}
 Let $V$ be a vector space over \(\F\) of dimension $n$. Fixing an $\F$-basis of $V$ yields an $\F$-algebra isomorphism
\[
\tau: M_n(\F)\longrightarrow\End_{\F}(V),
\]
which preserves rank. Hence $\tau(\C)\subseteq \End_{\F}(V)$, and the centraliser corresponds to
\begin{equation}\label{eq:mapcentraliser}
B = \tau\big(\cen(\C)\big)
        = \{\phi\in \End_{\F}(V) : \phi\circ a = a\circ \phi\ \text{for every } a\in \tau(\C)\}.
\end{equation}
By Proposition \ref{prop:centraliserisfield}, $\cen(\C)$ is a field containing a subfield isomorphic to $\F$, hence $B$ is a field with $\F\subseteq B\subseteq \End_{\F}(V)$. Suppose, for a contradiction, that $\cen(\C)\neq \F$, equivalently $[B:\F]=m\ge 2$. From \eqref{eq:mapcentraliser} we see that every $a\in \tau(\C)$ commutes with $B$, hence each $a$ is $B$-linear:
\[
a(\phi(\alpha))=\phi(a(\alpha))\qquad \text{for all } \phi\in B,\ \alpha\in V.
\]
Therefore, if we regard $V$ as a vector space over $B$ (of dimension $n/m$), the rank $\rk_{\F}(a)$ of $a$ over $\F$ satisfies the standard relation
\[
\rk_{\F}(a)= m\cdot \rk_{B}(a),
\]
where $\rk_{B}(a)$ denotes the rank of $a$ regarded as an $B$-linear map. This implies every possible rank over $\F$ attainable by codewords of $\C$ is a multiple of $m \geq 2$. Since $\C$ is MRD with $d(\C)<n$, it is known (see \cite[Lemma 2.1]{lunardon2018nuclei}) that $\C$ contains codewords having rank weights $n$ and $n-1$. This is impossible when all the rank weights are multiples of $m\ge 2$, because $n-1$ cannot be divisible by $m$. The contradiction shows that $m=1$, hence $\cen(\C)\cong \F$.
\end{proof}

\begin{remark}
It is worth pointing out that \Cref{cor:centraliserisfieldMRD} does not hold in general for MRD codes in $M_n(\F)$ with minimum distance equal to $n$. In fact, such codes correspond to semifields, and the centraliser $\cen(\C)$ is isomorphic to the right nucleus of the semifield associated with $\C$ (see \cite[Proposition 5]{sheekey2020new}). Therefore, besides having a representation as matrices in $M_n(\F)$, such codes may also admit representations over a field extension of $\F$. For example, consider the following code. Let $\sigma: x \mapsto x^q$ be the Frobenius automorphism of $\F_{q^4}$, and consider the skew polynomial ring $R = \F_{q^4}[x; \sigma]$. Define the code
    \[
    \C = \{ (a_0 + \delta a_1) + \gamma (a_2 + \delta a_3)x^2 \colon a_i \in \F_q \} \subseteq \frac{R}{R(x^4-1)} \cong M_4(\F_q),
    \]
    where $\delta \in \F_{q^4} \setminus \F_{q^2}$ and $\gamma \in \F_{q^2}$. Then $\C$ has dimension 4 over $\F_q$, and it can be shown that if $\delta^2$ is a nonsquare in $\F_{q^2}$ and $N_{\F_{q^2}/\F_q}(\gamma)$ is a nonsquare in $\F_q$, all elements of $\C$ have full rank (see \cite[Theorem 4.6]{ebert2009infinite}). Hence, $\C$ is an MRD code having minimum distance $n=4$. However, the centraliser of $\C$, which is isomorphic to the right nucleus of the corresponding semifield, is isomorphic to $\F_{q^2}$ (see \cite{ebert2009infinite} for further details).
\end{remark}

The idealisers, centraliser and the centre can be viewed as invariants of codes. For the centre and centraliser to be considered invariants, we need to assume that the identity matrix is contained in each code. Note that this is not a restrictive condition for MRD codes. Indeed, these codes always contain an invertible matrix, see e.g. \cite[Lemma 2.1]{lunardon2018nuclei}, and thus, up to equivalence, we can always assume that an MRD code contains the identity matrix.

\begin{proposition} [see \textnormal{\cite{lunardon2018nuclei} and \cite[Proposition 4]{sheekey2020new}}]
\label{prop:idequiv}
Suppose $\C$ and $\C'$ are two equivalent codes in $M_n(\F)$. Then
\[
|\lid(\C)| = |\lid(\C')| \mbox{ and}\quad |\rid(\C)| = |\rid(\C')|\]
and if both $\C$ and $\C'$ contain the identity, then
\[
\quad |\cen(\C)| = |\cen(\C')| \mbox{ and} \quad |Z(\C)| = |Z(\C')|.
\]
\end{proposition}

Relying on Proposition \ref{prop:idequiv}, for an MRD code $\C$ in $M_n(\F)$, we define the \textbf{nuclear parameters} of $\C$ as the tuple
\[
(\lvert \C \rvert, \lvert \lid(\C) \rvert,\lvert \rid(\C) \rvert,\lvert \cen(\C) \rvert,\lvert Z(\C) \rvert)
\]
Note that this definition also depends on the order of the matrices $n$ and the field $\F$.

For the families $S_{n,s,k}(\eta,\rho,F)$ and $D_{n,s,k}(\gamma,F)$, the nuclear parameters have been computed for $1 \leq k \leq n/2$. We reproduce the statements in the finite field case for the reader's convenience.

\begin{theorem} [\textnormal{\cite[Theorem 9]{sheekey2020new}}] \label{th:parconstrsheekey}
Let $q=p^e$, for some prime $p$. Assume that $1 \leq k \leq n/2$ and $sk>1$. Let $\C=S_{n,s,k}(\eta,\rho,F) \subseteq R_F \cong M_n(\F_{q^s})$ defined as in \eqref{eq:johncodes}. Assume that $\rho(y)=y^{p^h}$ and $\sigma(y)=y^{p^{ej}}$ for any $y \in \fqn$, with $(j,n)=1$. Let $\C'$ be any code equivalent to $\C$ containing the identity. If $\eta \neq 0$, then 
\[
\lid(\C') \cong \F_{p^{(ne,h)}}, \ \rid(\C') \cong  \F_{p^{(ne,ske-h)}}, \ \cen(\C')  \cong  \F_{p^{se}} \mbox{ and } Z(\C') \cong  \F_{p^{(e,h)}}
\]
If $\eta = 0$, then $S_{n,s,k}(0,\rho,F)=S_{n,s,k}(0,0,F)$ for all $\rho$, and 
\[
\lid(\C') \cong \F_{p^{ne}}, \ \rid(\C')  \cong  \F_{p^{ne}}, \ \cen(\C')  \cong  \F_{p^{se}}\mbox{ and } Z(\C') \cong  \F_{p^e}.
\]
\end{theorem}

\begin{theorem}  [see \textnormal{\cite[Theorem 6.2]{lobillo2025quotients}}]\label{th:parametersextensiontrombetti}
Assume that \(n\) is even, $1 \leq k \leq n/2$ and $sk\geq 3$. Let $\C=D_{n,s,k}(\gamma,F) \subseteq R_F \cong M_n(\F_{q^s})$ defined as in \eqref{eq:finiteextensiontrombzhou}. Then
\[
\lid(\C)  \cong \F_{q^{n/2}}, \ \rid(\C) \cong  \F_{q^{n/2}}, \ \cen(\C)  \cong  \F_{q^s} \mbox{ and }  Z(\C) \cong  \F_{q}.
\]
\end{theorem}


In the next results, we extend both theorems to the case $n/2 < k \leq n-1$. Let us start with the codes $S_{n,s,k}(\eta,\rho,F)$. 

\begin{theorem} \label{th:parconstrsheekey2}
Let $q=p^e$, for some prime $p$.  Assume that $1 \leq k \leq n-1$ and $sk>1$. Let $\C=S_{n,s,k}(\eta,\rho,F) \subseteq M_n(\F_{q^s})$ defined as in \eqref{eq:johncodes}. Assume that $\rho(y)=y^{p^h}$ and $\sigma(y)=y^{p^{ej}}$ for any $y \in \fqn$, with $\gcd(j,n)=1$. Let $\C'$ be any code equivalent to $\C$ containing the identity.  If $\eta \neq 0$, then 
\[
\lid(\C')  \cong \F_{p^{(ne,h)}},\ \rid(\C') \cong  \F_{p^{(ne,ske-h)}},\ \cen(\C') \cong \F_{p^{se}} \mbox{ and } Z(\C') \cong  \F_{p^{(e,h)}}
\]
If $\eta = 0$, then $S_{n,s,k}(0,\rho,F)=S_{n,s,k}(0,0,F)$ for all $\rho$, and 
\[
\lid(\C') \cong \F_{p^{ne}}, \rid(\C')  \cong  \F_{p^{ne}}, \cen(\C')  \cong  \F_{p^{se}} \mbox{ and } Z(\C') \cong  \F_{p^e}.
\]
\end{theorem}

\begin{proof}
By \cite[Proposition 4.2]{lunardon2018nuclei}, we know that, for a rank-metric code $\C \subseteq M_n(\F_{q^s})$, it holds 
\[
\lid(\C)=(\lid(\C^{\perp}))^{\top} \ \ \ \mbox{ and }\ \ \ \rid(\C)=(\rid(\C^{\perp}))^{\top}.
\]
Therefore, $\lid(\C)$ is isomorphic as a field to $\lid(\C^{\perp})$, and $\rid(\C)$ is isomorphic as a field to $\rid(\C^{\perp})$. Now, observe that if $k \geq n/2 + 1$, then $1 \leq n - k \leq n/2$. Thus, the assertion follows from the fact that the dual of a code $S_{n,s,k}(\eta, \rho, F)$ is $S_{n,s,n-k}(\rho^{-1}(\eta F_0), \rho^{-1}, \hat{F}) \subseteq R_{\hat{F}}$, as stated in part \emph{\ref{dualS})} of Proposition \ref{prop:determinantiondual} and by applying \Cref{th:parconstrsheekey}.
The proof for the centraliser and the centre immediately follows by Theorem  \ref{cor:centraliserisfieldMRD}
\end{proof}

Similarly, we have the following for the codes $D_{n,s,k}(\gamma,F)$. 

\begin{theorem}\label{th:parametersextensiontrombetti2}
Assume \(n\) to be even, that $1 \leq k < n$ and $sk\geq 3$. Let $\C=D_{n,s,k}(\gamma,F)$ defined as in \eqref{eq:finiteextensiontrombzhou}. Then 
\[
\lid(\C)  \cong \F_{q^{n/2}}, \ \rid(\C) \cong  \F_{q^{n/2}}, \  \cen(\C)  \cong  \F_{q^s} \mbox{ and } Z(\C) \cong  \F_{q}
\]
\end{theorem}

\begin{proof}
As in the proof of \Cref{th:parconstrsheekey2}, we recall that by \cite[Proposition 4.2]{lunardon2018nuclei}, $\lid(\C)$ is isomorphic as a field to $\lid(\C^{\perp})$, and $\rid(\C)$ is isomorphic as a field to $\rid(\C^{\perp})$. Thus, the assertion follows from the fact that the dual of a code $D_{n,s,k}(\gamma,F)$ is $D_{n,s,n-k}(\sigma^{sk}(\gamma),F)$, as proved in \emph{\ref{dualD})} of Proposition \ref{prop:determinantiondual} and by applying \Cref{th:parametersextensiontrombetti}.
The proof for the centraliser and the centre immediately follows by Theorem  \ref{cor:centraliserisfieldMRD}.
\end{proof}

In \Cref{tab:parameters}, we resume the known additive MRD codes in $M_n(\F)$ (with minimum distance less than $n$) together with their parameters, including the new results obtained in Theorems \ref{th:parconstrsheekey2} and \ref{th:parametersextensiontrombetti2}.

\begin{table}[ht]
\begin{footnotesize}
\begin{center}
    \begin{tabular}{|c|c|c|c|} 
        \hline
        & \textbf{Family} & \textbf{Nuclear parameters} & \textbf{Notes} \\ \hline
        I) & (Generalized) Gabidulin codes  & $(q^{nk},q^n,q^n,q,q)$ &  \\
        & (see \cite{delsarte1978bilinear,gabidulin1985theory,Gabidulins}) & $d=n-k+1$ & \\
        \hline
        II) & (Generalized)
Twisted Gabidulin & $\left(p^{nke},p^{(ne,h)},p^{(ne,ke-h)},p^{e},p^{(e,h)}\right)$ & $\rho(y)=y^{p^h}$, with $h < ne$ \\
 & codes  
  & $d=n-k+1$ & $\sigma(y)=y^{p^{ej}}$, with $(j,n)=1$  \\
        & (see \cite{sheekey2016new,otal2016additive,lunardon2018generalized}) &  & \\ \hline
        III) & Trombetti-Zhou codes  & $(q^{nk},q^{n/2},q^{n/2},q,q)$ & $q$ odd and $n$ even \\
        & (see \cite{trombetti2018new}) & $d=n-k+1$ & \\  \hline
        IV) & Csajb\'ok-Marino-Polverino-Zhou  & $(q^{nk},q^n,q^n,q,q)$ & $n = 7$ and $q$ odd or  \\ 
        & codes  & $d=n-k+1$ & $n = 8$, $q \equiv 1 \pmod{3}$ \\ 
        & (see \cite{csajbok2020mrd}) &  & $k \in \{3,4,5\}$ \\ \hline
        V) & Codes from scattered polynomials & $(q^{2n},q^n,\cdot,q,q)$ 
        & Some conditions on \\
         & (see \cite{longobardi2024scattered} and references therein)
         & 
         $d=n-1$ &  $n$ and $q$ required \\ \hline
        VI) & $S_{n,k,s}(\eta,\rho,F)$, with $\eta \neq 0$ 
         & $\left(p^{nske},p^{(ne,h)},p^{(ne,ske-h)},p^{se},p^{(e,h)}\right)$ & $\rho(y)=y^{p^h}$, with $h < ne$ \\ 
        & (see \cite{sheekey2020new}) & $d=n-k+1$ &  $\sigma(y)=y^{p^{ej}}$, with $(j,n)=1$ \\ \hline 
 VII) & $S_{n,k,s}(\eta,\rho,F)$, with $\eta = 0$ 
         & $\left(q^{nsk}, q^n,q^n,q^s,q\right)$ & \\ 
        & (see \cite{sheekey2020new}) & $d=n-k+1$ & \\ \hline
        VIII) & $D_{n,s,k}(\gamma,F)$ & $(q^{nsk},q^{n/2},q^{n/2},q^s,q)$ & $q$ odd and $n$ even \\ 
        & (see \cite{lobillo2025quotients}) & $d=n-k+1$& \\ \hline
    \end{tabular} 
\end{center}
\caption{Parameters of known MRD codes in $M_n(\F)$}
    \label{tab:parameters}
    \end{footnotesize}
\end{table}

\begin{theorem}The following hold: 
\begin{enumerate}
    \item The family $S_{n,s,k}(\eta,\rho,F)$ contains new MRD codes for $n/2 < k \leq n-1$, for all $n,s$ such that $\gcd(n,s)$ does not divide $e$, where $q = p^e$.
 \item The family $D_{n,s,k}(\gamma,F)$ contains new MRD codes for all $n/2 < k \leq n-1$ and $s \geq 3$ such that $n \nmid sk$. 
 \end{enumerate}

\end{theorem}

\begin{proof}
The nuclei, centralisers, and centre have been computed in \Cref{th:parconstrsheekey2} and \Cref{th:parametersextensiontrombetti2} for the codes $S_{n,s,k}(\eta, \rho, F)$ and $D_{n,s,k}(\gamma, F)$, respectively, including the case $n/2 < k \leq n - 1$. Then, using the same calculations as in the proofs of \cite[Theorem 9]{sheekey2020new} and \cite[Theorem 5.12]{lobillo2025quotients}, we obtain the assertion.
\end{proof}

\section*{Acknowledgments}
The research was partially supported by the Italian National Group for Algebraic and Geometric Structures and their Applications (GNSAGA - INdAM) and by Bando Galileo 2024 – G24-216. This research was partially supported by grant PID2023-149565NB-I00 funded by MICIU/AEI/ 10.13039/501100011033 and by FEDER, EU. Paolo Santonastaso is very grateful for the hospitality of the Departamento de \'Algebra of
Universidad de Granada, Spain, where he was visiting during the development of this research.

\vspace{0.5cm}

José G\'omez-Torrecillas, F. J. Lobillo, and Gabriel Navarro\\
Departamento de \'Algebra,\\
Facultad de Ciencias,
Universidad de Granada,\\
Av. Fuente Nueva s/n, 18071 Granada, Spain\\
{{\em \{gomezj,jlobillo,gnavarro\}@ugr.es}}

\medskip

Paolo Santonastaso\\
Dipartimento di Matematica e Fisica,\\
Universit\`a degli Studi della Campania ``Luigi Vanvitelli'',\\
I--\,81100 Caserta, Italy\\
{{\em paolo.santonastaso@unicampania.it}  \\
Dipartimento di Meccanica, Matematica e Management, \\
Politecnico di Bari, \\
Via Orabona 4, \\
70125 Bari, Italy \\
{\em paolo.santonastaso@poliba.it}
}

\end{document}